\documentclass[a4paper,12pt,reqno]{amsart}

\usepackage{soul,graphicx,amssymb,datetime,float,MnSymbol,currfile,tikz,hyperref,multirow}
\usepackage[round]{natbib}
\usepackage[foot]{amsaddr}
\usetikzlibrary{arrows}
\usetikzlibrary{decorations.markings}

\newtheorem{theorem}{Theorem}

\def\ci{\!\perp\!}

\newcommand{\comments}[1]{}

\tikzset{tt/.style={decoration={
  markings,
  mark=at position .485 with {\arrow{>}},
  mark=at position .515 with {\arrow{<}}},postaction={decorate}}}

\begin{document}

\title[]{Simple yet Sharp Sensitivity Analysis for Unmeasured Confounding}

\author{Jose M. Pe\~{n}a$^1$}
\address{$^1$Link\"oping University, Sweden.}
\email{jose.m.pena@liu.se}


\maketitle

\begin{abstract}
We present a method for assessing the sensitivity of the true causal effect to unmeasured confounding. The method requires the analyst to set two intuitive parameters. Otherwise, the method is assumption-free. The method returns an interval that contains the true causal effect, and whose bounds are arbitrarily sharp, i.e. practically attainable. We show experimentally that our bounds can be tighter than those obtained by the method of \citet{DingandVanderWeele2016a} which, moreover, requires to set one more parameter than our method. Finally, we extend our method to bound the natural direct and indirect effects when there are measured mediators and unmeasured exposure-outcome confounding.
\end{abstract}

\begin{figure}[t]
\centering
\begin{tabular}{c|c}
\begin{tikzpicture}[inner sep=1mm]
\node at (0,0) (E) {$E$};
\node at (2,0) (D) {$D$};
\node at (1,1) (U) {$U$};
\path[->] (E) edge (D);
\path[->] (U) edge (E);
\path[->] (U) edge (D);
\end{tikzpicture}
&
\begin{tikzpicture}[inner sep=1mm]
\node at (0,0) (E) {$E$};
\node at (1.5,0) (Z) {$Z$};
\node at (3,0) (D) {$D$};
\node at (1.5,1) (U) {$U$};
\path[->] (E) edge [bend right] (D);
\path[->] (E) edge (Z);
\path[->] (Z) edge (D);
\path[->] (U) edge (E);
\path[->] (U) edge (D);
\end{tikzpicture}
\end{tabular}\caption{Causal graphs where $U$ is unmeasured.}\label{fig:graphs}
\end{figure}
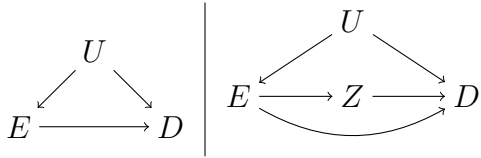

\section{Introduction}

Unmeasured confounding may bias the estimation of the true causal effect. One way to address this problem is through sensitivity analysis, i.e. reporting one or several intervals that include the true causal effect and whose bounds are functions of certain sensitivity parameter values provided by the analyst. These parameters are usually meant to quantify the association of the unmeasured confounders with the exposure and outcome. \citet{DingandVanderWeele2016a}, hereafter DV, proposed a method for sensitivity analysis that has received considerable attention, as evidenced by the survey by \citet{Blumetal.2020}. See also the follow-up works by \citet{DingandVanderWeele2016b}, \citet{VanderWeeleandDing2017}, \citet{VanderWeeleetal.2019} and \citet{Sjolander2020}. The latter shows that DV's interval bounds are not always sharp or attainable, i.e. logically possible.

In this work, we introduce a new method for sensitivity analysis. Our method requires the analyst to set two sensitivity parameters. This is one parameter less than DV's method. Otherwise, our method is assumption-free. We derive the feasible region for our parameters and show that, unlike DV's, our interval's bounds are arbitrarily sharp, i.e. practically attainable. Moreover, we show through simulations that our bounds can be tighter than DV's. This suggests that it may be advantageous to combine DV's and our method, by computing both sets of bounds and reporting the tightest of them.

Our sensitivity analysis method includes the parameter-free method proposed by \citet{Sjolander2020}, hereafter AS, as a special case. Specifically, AS' interval is the broadest our method can return: It is returned when the analyst chooses the least informative values for our sensitivity parameters. In other words, our bounds are always tighter than AS'. Like \citet{DingandVanderWeele2016a} and \citet{Sjolander2020}, we only consider binary outcomes. AS' bounds coincide with Manski's bounds for binary outcomes \citep{Manski1990}. Our bounds, on the other hand, can be seen as an adaptation of Manski's bounds for non-binary outcomes to binary outcomes. We elaborate further on this later. 

The rest of the paper is organized as follows. Section \ref{sec:RR} presents our method for sensitivity analysis of the risk ratio. Section \ref{sec:RD} extends it to the risk difference. Section \ref{sec:conditioning} extends our method to the risk ratio/difference conditioned or averaged over measured covariates. Sections \ref{sec:real} and \ref{sec:simulations} illustrate our method on real and simulated data. Section \ref{sec:mediation} considers the case where the effect of the exposure on the outcome is mediated by measured covariates, and our method is adapted to bound the natural direct and indirect effects under exposure-outcome confounding. Finally, Section \ref{sec:discussion} closes with some discussion.

\section{Bounds on the Risk Ratio}\label{sec:RR}

Consider the causal graph to the left in Figure \ref{fig:graphs}, where $E$ denotes the exposure, $D$ denotes the outcome, and $U$ denotes the set of unmeasured confounders. Let $E$ and $D$ be binary random variables. For simplicity, we assume that $U$ is a categorical random vector, but our results also hold for ordinal and continuous confounders. For simplicity, we treat $U$ as a categorical random variable whose levels are the Cartesian product of the levels of the components of the original $U$. We use upper-case letters to denote random variables, and the same letters in lower-case to denote their values. 

The causal graph to the left in Figure \ref{fig:graphs} represents a non-parametric structural equation model with independent errors, which defines a joint probability distribution $p(D,E,U)$. We make the usual positivity assumption that if $p(U=u) > 0$ then $p(E=e|U=u) > 0$, i.e. $E$ is not a deterministic function of $U$ and, thus, every individual in the subpopulations defined by the confounder can possibly be exposed or not \citep{HernanandRobins2020}. Then, the true risk ratio is defined as
\begin{equation}\label{eq:RRtrue}
RR^{true} = \frac{p(D_1=1)}{p(D_0=1)}
\end{equation}
where $D_e$ denotes the counterfactual outcome when the exposure is set to level $E=e$. Since there is no confounding besides $U$, we have that $D_e \ci E | U$ for all $e$ and, thus, we can write
\[
RR^{true} = \frac{\sum_u p(D=1 | E=1, U=u) p(U=u)}{\sum_u p(D=1 | E=0, U=u) p(U=u)}
\]
using first the law of total probability, then $D_e \ci E | U$ and, finally, the law of counterfactual consistency, i.e. $E=e \Rightarrow D_e = D$. This quantity is incomputable though. The observed risk ratio is defined as
\begin{equation}\label{eq:RRobs}
RR^{obs} = \frac{p(D=1 | E=1)}{p(D=1 | E=0)}
\end{equation}
which is computable. However, $RR^{true}$ and $RR^{obs}$ do not coincide in general. In this section, we give bounds on $RR^{true}$ in terms of the observed data distribution and two sensitivity parameters.

We start by noting that
\begin{align}\label{eq:counterfactual1}\nonumber
p(D_1=1) & = p(D_1=1 | E=1) p(E=1) + p(D_1=1 | E=0) p(E=0)\\
& = p(D=1 | E=1) p(E=1) + p(D_1=1 | E=0) p(E=0)
\end{align}
where the second equality follows from counterfactual consistency, and likewise
\begin{equation}\label{eq:counterfactual2}
p(D_0=1) = p(D_0=1 | E=1) p(E=1) + p(D=1 | E=0) p(E=0).
\end{equation}
If the analyst is able to confidently provide bounds on $p(D_1=1 | E=0)$ and $p(D_0=1 | E=1)$, then these can be used together with the observed data distribution to bound $RR^{true}$ via Equations \ref{eq:counterfactual1} and \ref{eq:counterfactual2}. However, bounding counterfactual probabilities directly may be difficult in some domains, else the analyst might bound Equation \ref{eq:RRtrue} directly. Therefore, we propose instead to bound them in terms of $p(D|E,U)$. Specifically,
\begin{align*}
p(D_1=1 | E=0) & = \sum_u p(D_1=1 | E=0, U=u) p(U=u | E=0)\\
& = \sum_u p(D=1 | E=1, U=u) p(U=u | E=0)\\
& \leq \max_{e,u} p(D=1 | E=e, U=u)
\end{align*}
where the second equality follows from $D_e \ci E | U$ for all $e$, and counterfactual consistency. Likewise,
\[
p(D_1=1 | E=0) \geq \min_{e,u} p(D=1 | E=e, U=u).
\]
We believe that bounding the counterfactual probabilities by specifying these maximum and minimum probabilities may be easier in some domains than bounding the counterfactual probabilities directly, e.g. in domains where the identity of the confounders is known but their values are not, or in domains where neither the identity nor the values of the confounders are known but where there is a consensus on conservative estimates of the maximum and minimum probabilities (we elaborate further on conservative estimates later). This is also the motivation behind DV's method, as it requires the analyst to quantify the relationship between $E$ and $U$ and the relationship between $U$ and $D$. See Appendix B for a recap of DV's sensitivity analysis.

Now, let us define
\[
M=\max_{e,u} p(D=1 | E=e, U=u)
\]
and
\[
m=\min_{e,u} p(D=1 | E=e, U=u).
\]
Then,
\begin{align}\nonumber
p(D=1, E=1) + p(E=0) m &\leq p(D_1=1)\\\label{eq:numerator}
& \leq p(D=1, E=1) + p(E=0) M
\end{align}
and
\begin{align}\nonumber
p(D=1, E=0) + p(E=1) m &\leq p(D_0=1)\\\label{eq:denominator}
& \leq p(D=1, E=0) + p(E=1) M.
\end{align}
Therefore, combining Equations \ref{eq:RRtrue}, \ref{eq:numerator} and \ref{eq:denominator}, we have that
\begin{equation}\label{eq:bounds}
LB \leq RR^{true} \leq UB
\end{equation}
where
\[
LB = \frac{p(D=1, E=1) + p(E=0) m}{p(D=1, E=0) + p(E=1) M}
\]
and
\[
UB = \frac{p(D=1, E=1) + p(E=0) M}{p(D=1, E=0) + p(E=1) m}.
\]

$M$ and $m$ are two sensitivity parameters whose values the analyst has to specify. By definition, these values must lie in the interval $[0,1]$ and $M \geq m$. The observed data distribution constrains the valid values further. To see it, note that
\[
p(D=1 | E=e) = \sum_u p(D=1 | E=e, U=u) p(U=u | E=e) \leq M
\]
for all $e$ and, likewise,
\[
p(D=1 | E=e) \geq m.
\]
Let us define
\[
M^*=\max_{e} p(D=1 | E=e)
\]
and
\[
m^*=\min_{e} p(D=1 | E=e).
\]
Then,
\[
M^* \leq M
\]
and
\[
m^* \geq m.
\]
We can thus define the feasible region for $M$ and $m$ as $M^* \leq M \leq 1$ and $0 \leq m \leq m^*$.

We close this section with some observations about the bounds $LB$ and $UB$. Theorem \ref{the:attainableRR} in Appendix A shows that the bounds are arbitrarily sharp, meaning that there is a distribution arbitrarily close to the observed data distribution (and, thus, indistinguishable in practice on the basis of a finite sample) for which $RR^{true}$ and $LB$ or $UB$ are arbitrarily close. So, the bounds are arbitrarily close to being logically possible. Recall that DV's bounds are not always logically possible \citep{Sjolander2020}. Note also that
\[
LB \leq \frac{p(E=1) M + p(E=0) m}{p(E=0) m + p(E=1) M} = 1
\]
and, likewise, $UB \geq 1$. Thus, our interval in Equation \ref{eq:bounds} always includes the null causal effect $RR^{true}=1$. Since our bounds are arbitrarily sharp, the null causal effect is practically attainable whenever our lower or upper bound equals 1. DV's interval does not necessarily include the null causal effect. However, when DV's lower or upper bound equals 1, the null causal effect is also attainable \citep{Sjolander2020}. The fact that our interval always includes the null causal effect and DV's may not does not mean that the latter are always closer to $RR^{true}$, as the experiments in Section \ref{sec:simulations} show.

DV's method requires the analyst to describe the relationship between $E$ and $U$ with two sensitivity parameters and the relationship between $U$ and $D$ with one parameter, whereas our method requires the analyst to describe only the relationship between $U$ and $D$ with two parameters. Therefore, our method has one parameter less than DV's method. As a consequence of not describing the relationship between $E$ and $U$, our interval always includes the null causal effect. In other words, the undescribed relationship between $E$ and $U$ may be so strong as to nullify the causal effect, i.e. explain away the observed association between $E$ and $D$. If we define the variation in $p(D=1|E,U)$ as $M-m$, then the pair of values $M=M^*$ and $m=m^*$ can be interpreted as the minimum variation in $p(D=1|E,U)$ that is needed to nullify the causal effect, regardless of the relationship between $E$ and $U$. This resembles the interpretation of DV's E-value, which is precisely defined as the minimum values of DV's parameters that nullify the causal effect. See Appendix B for a recap of DV's sensitivity analysis. There is, however, a major difference between both interpretations. DV's parameter values that are smaller than the E-value are insufficient to nullify the causal effect. There is no analogous result for our method, since we cannot consider less variation in $p(D=1|E,U)$ than $M^*-m^*$. In other words, no parameter values are insufficient to nullify the causal effect because our interval always includes the null causal effect in order to be valid regardless of the relationship between $E$ and $U$.

Note that $LB$ is decreasing in $M$ and increasing in $m$, while the opposite is true for $UB$. Therefore, using conservative estimates of $M$ and $m$ (i.e., a value larger than the true $M$ value and a value smaller than the true $m$ value) results in a wider interval that still contains $RR^{true}$. Note also that AS' bounds are a special case of our bounds when $M=1$ and $m=0$. See Appendix C for a recap of AS' sensitivity analysis. Therefore, our bounds are always tighter than AS', because our interval is widest when $M=1$ and $m=0$.

Note also that if $RR^{obs} \geq 1$, then $M^*=p(D=1 | E=1)$ and $m^*=p(D=1 | E=0)$ and, thus, $LB=1$ when we set $M=M^*$ and $m=m^*$. Likewise, $UB=1$ if $RR^{obs} \leq 1$ and we set $M=M^*$ and $m=m^*$.

Finally, our method requires to specify two probabilities whereas DV's method requires to specify three probability ratios. Which set of parameter values the analyst finds easier to specify may well depend on the domain under study. So, we will not argue in favour of either of them. However, we do want to describe a hypothetical scenario where setting our parameters may be easier. Let the causal graph to the left in Figure \ref{fig:graphs} model the effect of exercise ($E$) on cholesterol ($D$) when confounded by junior vs. senior age ($U$). The three random variables are binary, and $U$ is unmeasured. Suppose that, although the exact probabilities are unknown, it is known that $p(D=1|E=1,U=u) < p(D=1|E=0,U=u)$ for $u \in \{0,1\}$, and $p(D=1|E=e,U=1) > p(D=1|E=e,U=0)$ for $e \in \{0,1\}$. In other words, exercise decreases the probability of high cholesterol among juniors and seniors, and seniority increases the probability of high cholesterol among those who do not do and do exercise. In other words, these relationships show no qualitative effect modification by one covariate when keeping the other fixed. \citet{OgburnandVanderWeele2012a} argue that such relationships are common in epidemiology. Then, our sensitivity parameters simplify to $M=p(D=1|E=0,U=1)$ and $m=p(D=1|E=1,U=0)$, whereas DV's parameter $RR_{UD}$ reduces to
\[
RR_{UD} = \max \bigg\{ \frac{p(D=1|E=0,U=1)}{p(D=1|E=0,U=0)}, \frac{p(D=1|E=1,U=1)}{p(D=1|E=1,U=0)} \bigg\}.
\]
Suppose that most juniors exercise whereas most seniors do not. Then, the analyst may find easier setting our parameters than DV's, since the latter involves speculating about the rare cases of seniors who exercise and juniors who do not.

\section{Bounds on the Risk Difference}\label{sec:RD}

\citet{DingandVanderWeele2016a} and \citet{Sjolander2020} show that their bounds on the risk ratio can be adapted to bound the risk difference. Ours can also be adapted, as we show next. The true risk difference is defined as
\begin{align*}
RD^{true} &= p(D_1=1) - p(D_0=1)\\
&= \sum_u p(D=1 | E=1, U=u) p(U=u)\\
&- \sum_u p(D=1 | E=0, U=u) p(U=u).
\end{align*}
Therefore, combining Equations \ref{eq:numerator} and \ref{eq:denominator}, we have that
\begin{equation}\label{eq:bounds2}
LB^{\dagger} \leq RD^{true} \leq UB^{\dagger}
\end{equation}
with
\[
LB^{\dagger} = p(D=1, E=1) + p(E=0) m - p(D=1, E=0) - p(E=1) M
\]
and
\[
UB^{\dagger} = p(D=1, E=1) + p(E=0) M - p(D=1, E=0) - p(E=1) m.
\]

Theorem \ref{the:attainableRD} in Appendix A shows that our bounds for $RD^{true}$ are arbitrarily sharp. Finally, see Appendix D for an account of the relationship between our bounds and Manski's bounds \citep{Manski1990}.

\section{Conditioning and Averaging over Measured Covariates}\label{sec:conditioning}

So far, our results have concerned the whole population. However, they also hold for the subpopulation $C=c$ where $C$ is a set of measured covariates, provided that the causal graph to the left in Figure \ref{fig:graphs} is valid in that subpopulation. Note that $U$ previously represented all the confounders between $E$ and $D$, while it now represents all the confounders for the subpopulation $C=c$. To adapt our results to the subpopulation $C=c$, it suffices to condition on $C=c$ in all the previous expressions. For instance, the true risk ratio in the subpopulation $C=c$ is defined as
\begin{equation}\label{eq:RRtruec}
RR^{true}_c = \frac{p(D_1=1|C=c)}{p(D_0=1|C=c)}.
\end{equation}
Arguing as before, we have that
\[
LB_c \leq RR^{true}_c \leq UB_c
\]
where
\[
LB_c = \frac{p(D=1, E=1|C=c) + p(E=0|C=c) m_c}{p(D=1, E=0|C=c) + p(E=1|C=c) M_c}
\]
and
\[
UB_c = \frac{p(D=1, E=1|C=c) + p(E=0|C=c) M_c}{p(D=1, E=0|C=c) + p(E=1|C=c) m_c}
\]
with sensitivity parameters
\[
M_c=\max_{e,u} p(D=1 | E=e, U=u, C=c)
\]
and
\[
m_c=\min_{e,u} p(D=1 | E=e, U=u, C=c)
\]
and feasible region $M_c^* \leq M_c \leq 1$ and $0 \leq m_c \leq m^*_c$, where
\[
M^*_c=\max_{e} p(D=1 | E=e, C=c)
\]
and
\[
m^*_c=\min_{e} p(D=1 | E=e, C=c).
\]

Finally, we show that $RR^{true}$ can be bounded as
\begin{equation}\label{eq:boundsavg}
\min_c LB_c \leq RR^{true} \leq \max_c UB_c
\end{equation}
by averaging over $C$ in the numerator and denominator of Equation \ref{eq:RRtruec}. Specifically, assume for simplicity that $C$ is categorical. Then,
\begin{align*}
RR^{true} &= \frac{p(D_1=1)}{p(D_0=1)}\\
&= \frac{\sum_c p(D_1=1|C=c) p(C=c)}{\sum_c p(D_0=1|C=c) p(C=c)}\\
&= \frac{\sum_c RR^{true}_c p(D_0=1|C=c) p(C=c)}{\sum_c p(D_0=1|C=c) p(C=c)}
\end{align*}
which implies the desired result. Analogous results can be derived for the true risk difference. We omit the details. These derivations have previously been reported for DV's bounds \citep[eAppendix 2.5]{DingandVanderWeele2016a}. We include them here for completeness. Which of the bounds in Equations \ref{eq:bounds} and \ref{eq:boundsavg} is tightest depends of the sensitivity parameter values chosen. Of course, the analyst has to set more parameters in the latter case. In some domains, it may be reasonable to assume that some parameters are approximately constant across subpopulations.

\section{Real Data Example}\label{sec:real}

In this section, we illustrate our method for sensitivity analysis on the real data provided by \citet{HammondandHorn1958}. This work studied the association between smoking and mortality. \citet{DingandVanderWeele2016a} and \citet{Sjolander2020} also used these data to illustrate their methods. Specifically, we use the same data as \citeauthor{Sjolander2020}, which correspond to the association between smoking and total mortality, and for which $RR^{obs}=1.28$. See the work by \citeauthor{Sjolander2020} for a detailed description of the data. We extend the R code provided by \citeauthor{Sjolander2020} with our method. The resulting code is available \href{https://www.dropbox.com/s/4rxfux9tt95ldjz/sensitivityAnalysis3.R?dl=0}{here}.

\begin{table}[t]
\caption{Intervals for different values of the sensitivity parameters $M$ and $m$ in the feasible region $M^* \leq M \leq 1$ and $0 \leq m \leq m^*$. Recall that $M=\max_{e,u} p(D=1 | E=e, U=u)$, $m=\min_{e,u} p(D=1 | E=e, U=u)$, $M^*=\max_{e} p(D=1 | E=e)$, and $m^*=\min_{e} p(D=1 | E=e)$. In this case, $M^*=0.12$ and $m^* = 0.1$.}\label{tab:intervals}
\centering
\scriptsize
\begin{tabular}{c|l|ccccc|}
\multicolumn{2}{c}{} & \multicolumn{5}{c}{$M$}\\\cline{3-7}
\multicolumn{2}{c|}{} & 0.12 & 0.34 & 0.56 & 0.78 & 1\\\cline{2-7}
\multirow{5}{*}{$m$} & 0.1 & (1.00, 1.28) & (0.41, 1.76) & (0.26, 2.25) & (0.19, 2.73) & (0.15, 3.22)\\
& 0.07 & (0.96, 1.59) & (0.39, 2.19) & (0.25, 2.79) & (0.18, 3.40) & (0.14, 4.00)\\
& 0.05 & (0.91, 2.10) & (0.37, 2.89) & (0.23, 3.69) & (0.17, 4.49) & (0.13, 5.29)\\ 
& 0.02 & (0.87, 3.09) & (0.35, 4.27) & (0.22, 5.45) & (0.16, 6.63) & (0.13, 7.80)\\
& 0 & (0.82, 5.90) & (0.34, 8.15) & (0.21, 10.4) & (0.15, 12.6) & (0.12, 14.9)\\\cline{2-7}
\end{tabular}
\end{table}

\begin{figure}[t]
\centering
\includegraphics[scale=.7]{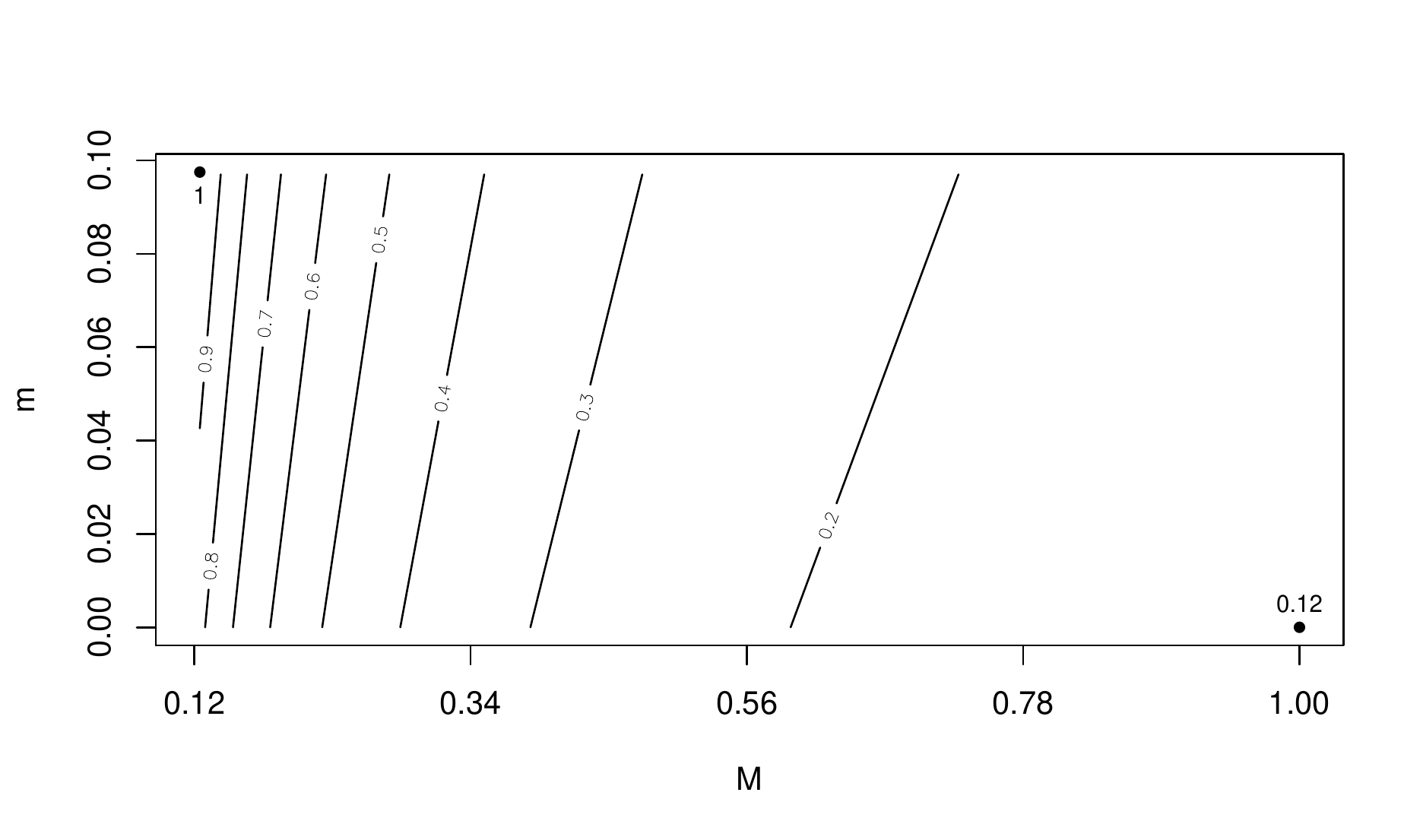}
\caption{Contour plot of the lower bound $LB$ as a function of the sensitivity parameters $M$ and $m$ in the feasible region $M^* \leq M \leq 1$ and $0 \leq m \leq m^*$. Recall that $M=\max_{e,u} p(D=1 | E=e, U=u)$, $m=\min_{e,u} p(D=1 | E=e, U=u)$, $M^*=\max_{e} p(D=1 | E=e)$, and $m^*=\min_{e} p(D=1 | E=e)$. In this case, $M^*=0.12$ and $m^* = 0.1$.}\label{fig:plots}
\end{figure}

Table \ref{tab:intervals} displays our interval for different $M$ and $m$ values in the feasible region $0.12 = M^* \leq M \leq 1$ and $0 \leq m \leq m^* = 0.1$. Figure \ref{fig:plots} complements the table with the contour plot of $LB$ as a function of $M$ and $m$. A similar plot can be produced for $UB$. An analyst can use the table and plot to determine a lower and/or upper bound for $RR^{true}$, given plausible values of $M$ and $m$. The table and plot illustrate some of the observations made before. Specifically, the null causal effect $RR^{true}=1$ is included in all the intervals. The lower bound of the intervals is decreasing in $M$ and increasing in $m$, where the opposite is true for the upper bound. The narrowest interval is achieved when $M=M^*$ and $m=m^*$, and the widest when $M=1$ and $m=0$. The lower bound of the narrowest interval is 1, because $RR^{obs} \geq 1$. Moreover, all the bounds in the table and plot are arbitrarily sharp (see Theorem \ref{the:attainableRR} in Appendix A).

\section{Simulated Experiments}\label{sec:simulations}

\citet{Sjolander2020} compares DV's and AS' bounds through simulations. In this section, we add our bounds to the comparison by extending the R code provided by \citeauthor{Sjolander2020}. The resulting code is available \href{https://www.dropbox.com/s/4rxfux9tt95ldjz/sensitivityAnalysis3.R?dl=0}{here}. Therefore, we follow \citeauthor{Sjolander2020} and consider a single binary confounder $U$, and generate distributions $p(D, E, U)$ from the model
\begin{align*}
p(E = 1) &= \text{expit}(\phi)\\
p(U = 1|E) &= \text{expit}(\alpha + \beta E)\\
p(D = 1|E, U) &= \text{expit}(\gamma + \delta E + \psi U)
\end{align*}
where $\text{expit}(x) = 1/(1+\exp(-x))$ is the inverse logit (a.k.a. logistic) function, and $\{\beta, \delta, \psi\}$ are independently distributed as $N(0,\sigma^2)$. We consider $\sigma=1, 3$ in the experiments. Note that $\sigma$ determines the probability of having confounding and causal effects of large magnitude. When this occurs, the sensitivity parameters may take large values, which results in wide intervals. In the experiments, the parameters $\{\phi, \alpha, \gamma\}$ are set to obtain certain marginal probabilities $\{p(U=1), p(E=1), p(D=1)\}$ specified below. For each combination of parameters, we generate 1000 distributions $p(D, E, U)$ from the model above.

Tables \ref{tab:simulation1} and \ref{tab:simulation2} summarize the results. Our bounds are more conservative than DV's but less than AS', as it can be seen in the columns $\tilde{\Delta}$, $\overline{\Delta}$ and $\Delta$. DV's bounds are usually tighter than AS' and ours, as it can be appreciated in the columns $\tilde{p}$ and $\overline{p}$. However, our bounds are tighter than DV's in a substantial fraction of the runs for some settings, e.g. see $\overline{p}$ for the upper bound with $\sigma=3$. We do not compare AS' and our bounds directly because, as discussed before, our bounds are always tighter than AS'. A plausible explanation of why DV's bounds are usually tighter than ours is that the former include information about the association between $E$ and $U$ through one of the sensitivity parameters, while the latter do not. A plausible explanation of why our bounds are sometimes tighter than DV's is the following: When the confounding and causal effects are large in magnitude (something that is more likely to occur with $\sigma=3$ than with $\sigma=1$), DV's parameters may take large values and, thus, DV's intervals may be too wide, since DV's bounds are not always sharp. This is less of a problem for our intervals as they cannot be arbitrarily wide, because our bounds are arbitrarily sharp. Therefore, if possible, it may be advantageous to compute both DV's and our bounds and report the tightest of them.

\begin{table}[t]
\caption{Simulation results with $\sigma = 1$. $\tilde{p}$ and $\overline{p}$ are the proportions of times that AS' bounds and our bounds are tighter than DV’s bounds, respectively. $\tilde{\Delta}$, $\overline{\Delta}$ and $\Delta$ are the mean absolute distance between the log of the bound and the log of the true risk ratio for AS' bounds, our bounds and DV’s bounds, respectively.}\label{tab:simulation1}
\centering
\scriptsize
\begin{tabular}{ccc|ccccc|ccccc}
&&&\multicolumn{5}{c}{lower bound}&\multicolumn{5}{c}{upper bound}\\
$p(U=1)$ & $p(E=1)$ & $p(D=1)$ & $\tilde{p}$ & $\overline{p}$ & $\tilde{\Delta}$ & $\overline{\Delta}$ & $\Delta$ & $\tilde{p}$ & $\overline{p}$ & $\tilde{\Delta}$ & $\overline{\Delta}$ & $\Delta$\\
\hline
0.05  & 0.05  & 0.05     & 0  & 0.10  & 3.67  & 0.79  & 0.15     & 0   & 0.16   & 3.04   & 0.71   & 0.17\\
0.05  & 0.05  & 0.20     & 0  & 0.12  & 3.18  & 0.62  & 0.13     & 0   & 0.16   & 1.71   & 0.59   & 0.15\\
0.05  & 0.20  & 0.05     & 0  & 0.10  & 3.23  & 0.73  & 0.13     & 0   & 0.15   & 3.10   & 0.72   & 0.17\\
0.05  & 0.20  & 0.20     & 0  & 0.12  & 2.22  & 0.59  & 0.13     & 0   & 0.14   & 1.76   & 0.58   & 0.13\\
0.20  & 0.05  & 0.05     & 0  & 0.06  & 3.68  & 0.78  & 0.14     & 0   & 0.13   & 3.05   & 0.67   & 0.17\\
0.20  & 0.05  & 0.20     & 0  & 0.08  & 3.18  & 0.65  & 0.12     & 0   & 0.15   & 1.71   & 0.59   & 0.16\\
0.20  & 0.20  & 0.05     & 0  & 0.04  & 3.26  & 0.77  & 0.14     & 0   & 0.10   & 3.07   & 0.69   & 0.17\\
0.20  & 0.20  & 0.20     & 0  & 0.06  & 2.23  & 0.63  & 0.13     & 0   & 0.14   & 1.71   & 0.55   & 0.14
\end{tabular}
\end{table}

\begin{table}[t]
\caption{Simulation results with $\sigma = 3$. $\tilde{p}$ and $\overline{p}$ are the proportions of times that AS' bounds and our bounds are tighter than DV’s bounds, respectively. $\tilde{\Delta}$, $\overline{\Delta}$ and $\Delta$ are the mean absolute distance between the log of the bound and the log of the true risk ratio for AS' bounds, our bounds and DV’s bounds, respectively.}\label{tab:simulation2}
\centering
\scriptsize
\begin{tabular}{ccc|ccccc|ccccc}
&&&\multicolumn{5}{c}{lower bound}&\multicolumn{5}{c}{upper bound}\\
$p(U=1)$ & $p(E=1)$ & $p(D=1)$ & $\tilde{p}$ & $\overline{p}$ & $\tilde{\Delta}$ & $\overline{\Delta}$ & $\Delta$ & $\tilde{p}$ & $\overline{p}$ & $\tilde{\Delta}$ & $\overline{\Delta}$ & $\Delta$\\
\hline
0.05 & 0.05 & 0.05 & 0.00 & 0.14 & 3.78 & 1.97 & 0.59 & 0.13  & 0.33  & 3.39  & 1.78  & 0.82\\
0.05 & 0.05 & 0.20 & 0.01 & 0.17 & 3.18 & 1.43 & 0.50 & 0.21  & 0.36  & 2.16  & 1.44  & 0.70\\
0.05 & 0.20 & 0.05 & 0.01 & 0.18 & 3.60 & 1.87 & 0.50 & 0.05  & 0.24  & 3.61  & 1.88  & 0.77\\
0.05 & 0.20 & 0.20 & 0.05 & 0.22 & 2.35 & 1.34 & 0.56 & 0.12  & 0.25  & 2.24  & 1.46  & 0.58\\
0.20 & 0.05 & 0.05 & 0.00 & 0.11 & 3.93 & 2.07 & 0.59 & 0.12  & 0.34  & 3.44  & 1.67  & 0.78\\
0.20 & 0.05 & 0.20 & 0.00 & 0.13 & 3.25 & 1.60 & 0.49 & 0.25  & 0.40  & 2.19  & 1.47  & 0.78\\
0.20 & 0.20 & 0.05 & 0.00 & 0.12 & 3.81 & 2.06 & 0.59 & 0.05  & 0.28  & 3.57  & 1.78  & 0.76\\
0.20 & 0.20 & 0.20 & 0.03 & 0.17 & 2.42 & 1.49 & 0.56 & 0.19  & 0.38  & 2.13  & 1.40  & 0.72\\
\end{tabular}
\end{table}

\begin{table}[t]
\caption{Simulation results with $\sigma = 1$, and parameter values that are 15 \% more conservative than the true values. $\tilde{p}$ and $\overline{p}$ are the proportions of times that AS' bounds and our bounds are tighter than DV’s bounds, respectively. $\tilde{\Delta}$, $\overline{\Delta}$ and $\Delta$ are the mean absolute distance between the log of the bound and the log of the true risk ratio for AS' bounds, our bounds and DV’s bounds, respectively.}\label{tab:simulation1conservative}
\centering
\scriptsize
\begin{tabular}{ccc|ccccc|ccccc}
&&&\multicolumn{5}{c}{lower bound}&\multicolumn{5}{c}{upper bound}\\
$p(U=1)$ & $p(E=1)$ & $p(D=1)$ & $\tilde{p}$ & $\overline{p}$ & $\tilde{\Delta}$ & $\overline{\Delta}$ & $\Delta$ & $\tilde{p}$ & $\overline{p}$ & $\tilde{\Delta}$ & $\overline{\Delta}$ & $\Delta$\\
\hline
0.05 & 0.05 & 0.05 & 0 & 0.07 & 3.67 & 0.95 & 0.23 & 0.00  & 0.13  & 3.04  & 0.85  & 0.26\\
0.05 & 0.05 & 0.20 & 0 & 0.08 & 3.18 & 0.78 & 0.21 & 0.01  & 0.12  & 1.71  & 0.73  & 0.23\\
0.05 & 0.20 & 0.05 & 0 & 0.06 & 3.23 & 0.88 & 0.21 & 0.00  & 0.12  & 3.10  & 0.86  & 0.26\\
0.05 & 0.20 & 0.20 & 0 & 0.08 & 2.22 & 0.74 & 0.21 & 0.00  & 0.10  & 1.76  & 0.72  & 0.21\\
0.20 & 0.05 & 0.05 & 0 & 0.00 & 3.68 & 0.94 & 0.22 & 0.00  & 0.10  & 3.05  & 0.81  & 0.26\\
0.20 & 0.05 & 0.20 & 0 & 0.01 & 3.18 & 0.80 & 0.20 & 0.01  & 0.10  & 1.71  & 0.73  & 0.24\\
0.20 & 0.20 & 0.05 & 0 & 0.01 & 3.26 & 0.92 & 0.22 & 0.00  & 0.08  & 3.07  & 0.83  & 0.25\\
0.20 & 0.20 & 0.20 & 0 & 0.01 & 2.23 & 0.78 & 0.21 & 0.00  & 0.09  & 1.71  & 0.69  & 0.23
\end{tabular}
\end{table}

\begin{table}[t]
\caption{Simulation results with $\sigma = 3$, and parameter values that are 15 \% more conservative than the true values. $\tilde{p}$ and $\overline{p}$ are the proportions of times that AS' bounds and our bounds are tighter than DV’s bounds, respectively. $\tilde{\Delta}$, $\overline{\Delta}$ and $\Delta$ are the mean absolute distance between the log of the bound and the log of the true risk ratio for AS' bounds, our bounds and DV’s bounds, respectively.}\label{tab:simulation2conservative}
\centering
\scriptsize
\begin{tabular}{ccc|ccccc|ccccc}
&&&\multicolumn{5}{c}{lower bound}&\multicolumn{5}{c}{upper bound}\\
$p(U=1)$ & $p(E=1)$ & $p(D=1)$ & $\tilde{p}$ & $\overline{p}$ & $\tilde{\Delta}$ & $\overline{\Delta}$ & $\Delta$ & $\tilde{p}$ & $\overline{p}$ & $\tilde{\Delta}$ & $\overline{\Delta}$ & $\Delta$\\
\hline
0.05 & 0.05 & 0.05 & 0.01 & 0.14 & 3.78 & 2.10 & 0.71 & 0.14  & 0.30  & 3.39  & 1.91  & 0.94\\
0.05 & 0.05 & 0.20 & 0.01 & 0.15 & 3.18 & 1.56 & 0.61 & 0.27  & 0.37  & 2.16  & 1.55  & 0.82\\
0.05 & 0.20 & 0.05 & 0.01 & 0.17 & 3.60 & 2.01 & 0.62 & 0.06  & 0.23  & 3.61  & 2.02  & 0.89\\
0.05 & 0.20 & 0.20 & 0.06 & 0.21 & 2.35 & 1.45 & 0.67 & 0.14  & 0.25  & 2.24  & 1.58  & 0.69\\
0.20 & 0.05 & 0.05 & 0.00 & 0.07 & 3.93 & 2.20 & 0.70 & 0.14  & 0.32  & 3.44  & 1.81  & 0.90\\
0.20 & 0.05 & 0.20 & 0.00 & 0.08 & 3.25 & 1.72 & 0.61 & 0.29  & 0.39  & 2.19  & 1.59  & 0.90\\
0.20 & 0.20 & 0.05 & 0.01 & 0.09 & 3.81 & 2.19 & 0.70 & 0.06  & 0.25  & 3.57  & 1.91  & 0.88\\
0.20 & 0.20 & 0.20 & 0.04 & 0.13 & 2.42 & 1.61 & 0.68 & 0.21  & 0.35  & 2.13  & 1.52  & 0.84
\end{tabular}
\end{table}

The experiments above assume that the analyst knows the true sensitivity parameter values, which is rarely the case. More realistic experiments make use of parameter values that are more conservative than the true values. Following \citeauthor{Sjolander2020}, we repeat the experiments above using DV's parameter values that are 15 \% larger than the true values. Likewise, we use $M$ and $m$ values that are, respectively, 15 \% larger and 15 \% smaller than the true values. As discussed in Section \ref{sec:RR}, this should make our bounds more conservative. Likewise for DV's bounds \citep{Sjolander2020}. Tables \ref{tab:simulation1conservative} and \ref{tab:simulation2conservative} report exactly how much more conservative the bounds become. Specifically, the columns $\overline{\Delta}$ and $\Delta$ show that the bounds are slightly more conservative than before but not much, which leads us to conclude that neither DV’s nor our bounds are overly sensitive to conservative estimates of the parameters. Still, our bounds are tighter than DV's in a considerable fraction of the runs, as shown in the column $\overline{p}$ in Tables \ref{tab:simulation1conservative} and \ref{tab:simulation2conservative}. That $\overline{p}$ is slightly smaller in these tables than in Tables \ref{tab:simulation1} and \ref{tab:simulation2} can arguably be attributed to the experimental setting being advantageous for DV's bounds. Our argument is as follows. One of DV's sensitivity parameters (see Appendix B) is
\[
RR_{UD} = \max_e \frac{\max_u p(D=1 | E=e, U=u)}{\min_u p(D=1 | E=e, U=u)}.
\]
Then, $RR_{UD} \leq \frac{M}{m}$. Consider those simulations where $RR_{UD} = \frac{M}{m}$. In those simulations, $M$ and $m$ are replaced by the conservative estimates $1.15 \cdot M$ and $0.85 \cdot m$ to compute our bounds. So, $\frac{M}{m}$ corresponds to $\frac{1.15}{0.85} \frac{M}{m}$ in those simulations. Thus, one may argue that $RR_{UD}$ should be replaced by $\frac{1.15}{0.85} RR_{UD}$ to compute DV's bounds in those simulations. However, it is replaced by the less conservative $1.15 \cdot RR_{UD}$. Alternatively, one may argue that using $1.15 \cdot RR_{UD}$ in those simulations corresponds to using $1.15 \cdot M$ and $m$, instead of the more conservative $0.85 \cdot m$.

\section{Bounds for Mediation}\label{sec:mediation}

So far, we have focused on bounding the total causal effect of the exposure $E$ on the outcome $D$. However, if the relationship between $E$ and $D$ is mediated by some measured covariates $Z$, then it may also be interesting to bound the natural direct and indirect effects. This section adapts our sensitivity analysis method for this purpose. Specifically, we have previously considered the causal graph to the left in Figure \ref{fig:graphs}. We now consider the refined causal graph to the right in the figure. Note that there is unmeasured exposure-outcome confounding ($U$), but no unmeasured exposure-mediator or mediator-outcome confounding.\footnote{That we assume no mediator-outcome confounding does not mean that we assume randomization of the mediator, because that is infeasible.} We defer the study of the latter two cases to a future work. \citet{DingandVanderWeele2016b} adapted DV's method to bound the natural direct and indirect effects under unmeasured mediator-outcome confounding but no unmeasured exposure-mediator or exposure-outcome confounding, which is always justified when the exposure is randomly assigned. In non-randomized studies like our work, the type of unmeasured confounding assumed can only be justified by substantive knowledge.

As before, let $D$ and $E$ be binary, and $Z$ and $U$ be categorical. The true natural direct effect is defined as
\[
RR^{true}_{NDE} = \frac{p(D_{1 Z_0}=1)}{p(D_{0 Z_0}=1)}
\]
where $Z_{e'}$ denotes the counterfactual value of the mediator when the exposure is set to level $E=e'$, and $D_{e Z_{e'}}$ denotes the counterfactual outcome when the exposure and mediator are set to levels $E=e$ and $Z_{e'}$, respectively. Note that the causal graph to the right in Figure \ref{fig:graphs} implies cross-world counterfactual independence, i.e. $D_{ez} \ci Z_{e'}$ for all $e$, $e'$ and $z$. \citet{Pearl2001} shows that we can then write
\[
RR^{true}_{NDE} = \frac{\sum_z p(D_{1 z}=1) p(Z_0=z)}{\sum_z p(D_{0 z}=1) p(Z_0=z)}.
\]
Since there is no exposure-mediator confounding, we have that $Z_{e} \ci E$ for all $e$ and, thus, we can write
\begin{equation}\label{eq:RRtrueNDE}
RR^{true}_{NDE} = \frac{\sum_z p(D_{1 z}=1) p(Z=z|E=0)}{\sum_z p(D_{0 z}=1) p(Z=z|E=0)}
\end{equation}
using the law of counterfactual consistency. Since there is no unmeasured confounding besides $U$, we have that $D_{ez} \ci (E, Z) | U$ for all $e$ and $z$ and, thus, we can write
\[
RR^{true}_{NDE} = \frac{\sum_z \sum_u p(D=1|E=1, Z=z, U=u) p(U=u) p(Z=z|E=0)}{\sum_z \sum_u p(D=1|E=0, Z=z, U=u) p(U=u) p(Z=z|E=0)}
\]
using first the law of total probability, then $D_{ez} \ci (E, Z) | U$ and, finally, the law of counterfactual consistency. For the previous quantity to be well-defined, we make the positivity assumption that if $p(U=u) > 0$ then $p(E=e|U=u) > 0$ and $p(Z=z|E=e) > 0$. Still, the previous quantity is incomputable. We give below bounds on $RR^{true}_{NDE}$ in terms of the observed data distribution and two sensitivity parameters.

We start by noting that
\begin{align*}
p(D_{1z}=1) & = p(D_{1z}=1 | E=1) p(E=1) + p(D_{1z}=1 | E=0) p(E=0)\\
& = p(D_{1z}=1 | E=1, Z=z) p(E=1) + p(D_{1z}=1 | E=0, Z=z) p(E=0)\\
& = p(D=1 | E=1, Z=z) p(E=1) + p(D_{1z}=1 | E=0, Z=z) p(E=0)
\end{align*}
where the second equality follows $D_{e,z} \ci Z | E$ for all $e$ and $z$, and the third from counterfactual consistency. Moreover,
\begin{align*}
p(D_{1z}=1 | E=0, Z=z) & = \sum_u p(D_{1z}=1 | E=0, Z=z, U=u) p(U=u | E=0, Z=z)\\
& = \sum_u p(D=1 | E=1, Z=z, U=u) p(U=u | E=0, Z=z)\\
& \leq \max_{e,z,u} p(D=1 | E=e, Z=z, U=u)
\end{align*}
where the second equality follows from $D_{ez} \ci (E,Z) | U$ for all $e$ and $z$, and counterfactual consistency. Likewise,
\[
p(D_{1z}=1 | E=0, Z=z) \geq \min_{e,z,u} p(D=1 | E=e, Z=z, U=u).
\]
Now, let us define
\[
M=\max_{e,z,u} p(D=1 | E=e, Z=z, U=u)
\]
and
\[
m=\min_{e,z,u} p(D=1 | E=e, Z=z, U=u).
\]
Then,
\begin{align}\nonumber
& p(D=1 | E=1, Z=z) p(E=1) + p(E=0) m \leq p(D_{1z}=1)\\\label{eq:numeratorNDE}
& \leq p(D=1 | E=1, Z=z) p(E=1) + p(E=0) M.
\end{align}
Likewise,
\[
p(D_{0z}=1) = p(D_{0z}=1 | E=1, Z=z) p(E=1) + p(D=1 | E=0, Z=z) p(E=0)
\]
and, thus,
\begin{align}\nonumber
& p(D=1 | E=0, Z=z) p(E=0) + p(E=1) m \leq p(D_{0z}=1)\\\label{eq:denominatorNDE}
& \leq p(D=1 | E=0, Z=z) p(E=0) + p(E=1) M.
\end{align}
Therefore, combining Equations \ref{eq:RRtrueNDE}-\ref{eq:denominatorNDE}, we have that
\begin{equation}\label{eq:boundsNDE}
LB_{NDE} \leq RR^{true}_{NDE} \leq UB_{NDE}
\end{equation}
where
\[
LB_{NDE} = \frac{\sum_z [ p(D=1 | E=1, Z=z) p(E=1) + p(E=0) m ] p(Z=z|E=0)}{\sum_z [ p(D=1 | E=0, Z=z) p(E=0) + p(E=1) M ] p(Z=z|E=0)}
\]
and
\[
UB_{NDE} = \frac{\sum_z [ p(D=1 | E=1, Z=z) p(E=1) + p(E=0) M ] p(Z=z|E=0)}{\sum_z [ p(D=1 | E=0, Z=z) p(E=0) + p(E=1) m ] p(Z=z|E=0)}.
\]
As before, $M$ and $m$ are two sensitivity parameters whose values the analyst has to set. The feasible region for these parameters is $M^* \leq M \leq 1$ and $0 \leq m \leq m^*$ where
\[
M^*=\max_{e, z} p(D=1 | E=e, Z=z)
\]
and
\[
m^*=\min_{e, z} p(D=1 | E=e, Z=z).
\]

Finally, the true natural indirect effect is defined as
\[
RR^{true}_{NIE} = \frac{p(D_{1 Z_1}=1)}{p(D_{1 Z_0}=1)}.
\]
Under the assumptions above, we can write
\[
RR^{true}_{NDE} = \frac{\sum_z \sum_u p(D=1|E=1, Z=z, U=u) p(U=u) p(Z=z|E=1)}{\sum_z \sum_u p(D=1|E=1, Z=z, U=u) p(U=u) p(Z=z|E=0)}.
\]
Repeating the reasoning above, we can bound the incomputable $RR^{true}_{NIE}$ in terms of the observed data distribution and the sensitivity parameters $M$ and $m$ as follows:
\[
LB_{NIE} \leq RR^{true}_{NIE} \leq UB_{NIE}
\]
where
\[
LB_{NIE} = \frac{\sum_z [ p(D=1 | E=1, Z=z) p(E=1) + p(E=0) m ] p(Z=z|E=1)}{\sum_z [ p(D=1 | E=1, Z=z) p(E=1) + p(E=0) M ] p(Z=z|E=0)}
\]
and
\[
UB_{NIE} = \frac{\sum_z [ p(D=1 | E=1, Z=z) p(E=1) + p(E=0) M ] p(Z=z|E=1)}{\sum_z [ p(D=1 | E=1, Z=z) p(E=1) + p(E=0) m ] p(Z=z|E=0)}.
\]

Theorem \ref{the:attainableNDE} in Appendix A shows that our bounds for $RR^{true}_{NDE}$ are arbitrarily sharp. One can prove the result for $RR^{true}_{NIE}$ in much the same way. One can also extend our bounds to the risk difference scale, and to conditioning and averaging over measured covariates. We omit the details.

\section{Discussion}\label{sec:discussion}

In this work, we have introduced a new method for assessing the sensitivity of the risk ratio to unmeasured confounding. Our method requires the analyst to set two intuitive parameters. Otherwise, our method makes no parametric or modelling assumptions about the causal relationships under consideration. The resulting bounds of the risk ratio are guaranteed to be arbitrarily sharp. Moreover, we have adapted our method to bound the risk difference and the natural direct and indirect effects, even when conditioning or averaging over measured covariates. We have illustrated our method on real data, and shown via simulations that it can produce tighter bounds than DV's method \citep{DingandVanderWeele2016a}. Therefore, it may be a good practice to apply both methods and report the tightest bounds obtained. This presumes that the analyst knows the true sensitivity parameter values for both methods or, more realistically, some conservative estimates of them. Otherwise, there is no reason to prefer the tightest bounds, as they may exclude the true risk value. For the same reason, if the analyst can confidently produce conservative estimates for one method but not for the other, then it may be sensible to just use the former method. Recall that our method requires to estimate two probabilities, whereas DV's method requires to estimate three probability ratios. For which method the analyst can confidently produce conservative estimates may well depend on the domain under study. Therefore, we believe that our method and DV's complement each other, combined or separately.

Our bounds on the natural direct and indirect effects assume that there is only unmeasured exposure-outcome confounding. In the future, we would like to extend them to other types of confounding.

\section*{Acknowledgements}
We thank the Associate Editor and Reviewers for their comments, which helped us to improve our work. We gratefully acknowledge financial support from the Swedish Research Council (ref. 2019-00245).

\bibliographystyle{plainnat}
\bibliography{sensitivityAnalysis}

\section*{Appendix A: Theorems}

\begin{theorem}\label{the:attainableRR}
The bounds in Equation \ref{eq:bounds} are arbitrarily sharp.
\end{theorem}

\begin{proof}
Let the set $\{M', m', p'(D,E)\}$ represent the observed data distribution and sensitivity parameter values at hand. We assume that $M'$ and $m'$ belong to the feasible region. To show that the lower bound is arbitrarily sharp, we construct a distribution $p(D,E,U)$ that marginalizes to the set $\{M, m, p(D,E)\}$ such that (i) $\{M, m, p(D,E)\}$ and $\{M', m', p'(D,E)\}$ are arbitrarily close, and (ii) $LB$ and $RR^{true}$ are arbitrarily close.
\begin{itemize}
    \item Let $p(E)=p'(E)$.
    \item Let $U$ be binary with $p(U=1|E=1)=p(U=0|E=0)=1-\epsilon$ where $\epsilon$ is an arbitrary number such that $0<\epsilon<1$. The purpose of $\epsilon$ is to ensure that the positivity assumption holds.
    \item Let 
\begin{align*}
p(D=1|E=1,U=1)&=p'(D=1|E=1)\\
p(D=1|E=1,U=0)&=m'\\
p(D=1|E=0,U=1)&=M'\\
p(D=1|E=0,U=0)&=p'(D=1|E=0).
\end{align*}   
\end{itemize}
Note that $M' \geq \max_e p'(D=1|E=e)$ and $m' \leq \min_e p'(D=1|E=e)$, because $M'$ and $m'$ belong to the feasible region. Then, $M=M'$ and $m=m'$. Note also that
\[
p(U=1)=\sum_e p(U=1|E=e) p(E=e)=\epsilon \: p(E=0) + (1-\epsilon) \: p(E=1)
\]
and, thus, $p(U=1)$ can be made arbitrarily close to $p(E=1)$ by choosing $\epsilon$ sufficiently close to 0. Likewise,
\begin{align*}
p(D=1|E=1) &= \sum_u p(D=1|E=1,U=u) p(U=u|E=1)\\
&=\epsilon \: m' + (1-\epsilon) \: p'(D=1|E=1)
\end{align*}
and, thus, $p(D=1|E=1)$ can be made arbitrarily close to $p'(D=1|E=1)$ by choosing $\epsilon$ sufficiently close to 0. Likewise for $p(D=1|E=0)$ and $p'(D=1|E=0)$. Therefore, $LB$ and $RR^{true}$ can be made arbitrarily close by choosing $\epsilon$ sufficiently close to 0:
\begin{align*}
LB &= \frac{p(D=1 | E=1) p(E=1) + p(E=0) m}{p(D=1 | E=0) p(E=0) + p(E=1) M}\\
&\approx \frac{p'(D=1 | E=1) p(U=1) + p(U=0) m'}{p'(D=1 | E=0) p(U=0) + p(U=1) M'}\\
&= \frac{p(D=1 | E=1, U=1) p(U=1) + p(U=0) p(D=1 | E=1, U=0)}{p(D=1 | E=0,U=0) p(U=0) + p(U=1) p(D=1 | E=0, U=1)}\\
&= RR^{true}.
\end{align*} 

That the upper bound is arbitrarily sharp can be proven analogously, after the swap $p(D=1|E=1,U=0)=M'$ and $p(D=1|E=0,U=1)=m'$.
\end{proof}

\begin{theorem}\label{the:attainableRD}
The bounds in Equation \ref{eq:bounds2} are arbitrarily sharp.
\end{theorem}

\begin{proof}
Consider the same distribution as in the proof of Theorem \ref{the:attainableRR}. To show that the lower bound is arbitrarily sharp, it suffices to note that $LB^{\dagger}$ and $RD^{true}$ can be made arbitrarily close by choosing $\epsilon$ sufficiently close to 0:
\begin{align*}
LB^{\dagger} &= [m - p(D=1 | E=0) ] p(E=0) + p(E=1) [ p(D=1 | E=1) - M ]\\
&\approx [ m' - p'(D=1 | E=0) ] p(U=0) + p(U=1) [ p'(D=1 | E=1) - M']\\
&= [ p(D=1 | E=1, U=0) - p(D=1 | E=0, U=0) ] p(U=0)\\
&+ p(U=1) [ p(D=1 | E=1, U=1) - p(D=1 | E=0, U=1)]\\
&= RD^{true}.
\end{align*} 
That the upper bound is arbitrarily sharp can be proven analogously, after the swap $p(D=1|E=1,U=0)=M'$ and $p(D=1|E=0,U=1)=m'$.
\end{proof}

\begin{theorem}\label{the:attainableNDE}
The bounds in Equation \ref{eq:boundsNDE} are arbitrarily sharp.
\end{theorem}

\begin{proof}
Let the set $\{M', m', p'(D,E,Z)\}$ represent the observed data distribution and sensitivity parameter values at hand. We assume that $M'$ and $m'$ belong to the feasible region. To show that the lower bound is arbitrarily sharp, we construct a distribution $p(D,E,U,Z)$ that marginalizes to the set $\{M, m, p(D,E,Z)\}$ such that (i) $\{M, m, p(D,E,Z)\}$ and $\{M', m', p'(D,E,Z)\}$ are arbitrarily close, and (ii) $LB_{NDE}$ and $RR^{true}_{NDE}$ are arbitrarily close.
\begin{itemize}
    \item Let $p(E,Z)=p'(E,Z)$.
    \item Let $U$ be binary with $p(U=1|E=1, Z=z)=p(U=0|E=0, Z=z)=1-\epsilon$ for all $z$, where $\epsilon$ is an arbitrary number such that $0<\epsilon<1$. The purpose of $\epsilon$ is to ensure that the positivity assumption holds.
    \item For all $z$, let
\begin{align*}
p(D=1|E=1,Z=z,U=1)&=p'(D=1|E=1,Z=z)\\
p(D=1|E=1,Z=z,U=0)&=m'\\
p(D=1|E=0,Z=z,U=1)&=M'\\
p(D=1|E=0,Z=z,U=0)&=p'(D=1|E=0,Z=z).
\end{align*}   
\end{itemize}
Note that $M' \geq \max_{e,z} p'(D=1|E=e,Z=z)$ and $m' \leq \min_{e,z} p'(D=1|E=e,Z=z)$, because $M'$ and $m'$ belong to the feasible region. Then, $M=M'$ and $m=m'$. Note also that
\[
p(U=1)=\sum_e p(U=1|E=e) p(E=e)=\epsilon \: p(E=0) + (1-\epsilon) \: p(E=1)
\]
and, thus, $p(U=1)$ can be made arbitrarily close to $p(E=1)$ by choosing $\epsilon$ sufficiently close to 0. Likewise,
\begin{align*}
p(D=1|E=1,Z=z) &= \sum_u p(D=1|E=1,Z=z,U=u) p(U=u|E=1,Z=z)\\
&=\epsilon \: m' + (1-\epsilon) \: p'(D=1|E=1,Z=z)
\end{align*}
and, thus, $p(D=1|E=1,Z=z)$ can be made arbitrarily close to $p'(D=1|E=1,Z=z)$ by choosing $\epsilon$ sufficiently close to 0. Likewise for $p(D=1|E=0,Z=z)$ and $p'(D=1|E=0,Z=z)$. Therefore, $LB_{NDE}$ and $RR^{true}_{NDE}$ can be made arbitrarily close by choosing $\epsilon$ sufficiently close to 0:
\begin{align*}
LB_{NDE} &= \frac{\sum_z [ p(D=1| E=1, Z=z) p(E=1) + p(E=0) m ] p(Z=z|E=0)}{\sum_z [ p(D=1| E=0, Z=z) p(E=0) + p(E=1) M ] p(Z=z|E=0)}\\
&\approx \frac{\sum_z [ p'(D=1| E=1, Z=z) p(U=1) + p(U=0) m' ] p(Z=z|E=0)}{\sum_z [ p'(D=1| E=0, Z=z) p(U=0) + p(U=1) M' ] p(Z=z|E=0)}\\
&= \frac{\sum_z [ p(D=1| E=1, Z=z, U=1) p(U=1) + p(U=0) m' ] p(Z=z|E=0)}{\sum_z [ p(D=1| E=0, Z=z, U=0) p(U=0) + p(U=1) M' ] p(Z=z|E=0)}\\
&= RR^{true}_{NDE}
\end{align*}
because $p(D=1|E=1,Z=z,U=0)=m'$ and $p(D=1|E=0,Z=z,U=1)=M'$.

That the upper bound is arbitrarily sharp can be proven analogously, after the swap $p(D=1|E=1,Z=z,U=0)=M'$ and $p(D=1|E=0,Z=z,U=1)=m'$.
\end{proof}

\section*{Appendix B: DV's Sensitivity Analysis}

\citet{DingandVanderWeele2016a} prove that $RR^{true}$ can be bounded in terms of $RR^{obs}$ and the sensitivity parameters $RR_{UD}$, $RR_{E0U}$ and $RR_{E1U}$, whose values the analyst has to specify. Specifically and using the notation by \citet{Sjolander2020} for conciseness, \citeauthor{DingandVanderWeele2016a} prove that
\[
RR^{obs} / BF_1 \leq RR^{true} \leq RR^{obs} BF_0
\]
with
\[
BF_e = \frac{RR^{EeU} RR_{UD}}{RR_{EeU} + RR_{UD} -1}
\]
and where
\[
RR_{UD} = \max_e \frac{\max_u p(D=1 | E=e, U=u)}{\min_u p(D=1 | E=e, U=u)}
\]
and
\[
RR_{EeU} = \max_u \frac{p(U=u|E=e)}{p(U=u|E=1-e)}.
\]
Moreover, assume that $RR^{obs} > 1$. Otherwise, consider $1/RR^{obs}$. \citet{VanderWeeleandDing2017} define the E-value as
\[
\text{E-value} = \min_{\{RR_{E1U}, RR_{UD}\} : BF_1 \geq RR^{obs}} \max \{RR_{E1U}, RR_{UD}\}
\]
and show that
\[
\text{E-value} = RR^{obs} + \sqrt{RR^{obs} (RR^{obs} -1)}.
\]

\section*{Appendix C: AS' Sensitivity Analysis}

\citet{Sjolander2020} proposes the following parameter-free bounds of $RR^{true}$ in terms of $RR^{obs}$ and the observed data distribution:
\[
RR^{obs} / \widetilde{BF}_1 \leq RR^{true} \leq RR^{obs} \widetilde{BF}_0
\]
where
\[
\widetilde{BF}_e = \frac{p(D=1|E=1-e) p(E=1-e)+p(E=e)}{p(D=1|E=1-e) p(E=e)}.
\]
\citeauthor{Sjolander2020} also adapts the previous bounds to the risk difference scale:
\begin{equation}\label{eq:bounds2AS}
RD^{obs} - \widetilde{BF}_1^{\dagger} \leq RD^{true} \leq RD^{obs} + \widetilde{BF}_0^{\dagger}
\end{equation}
where
\[
RD^{obs} = p(D=1|E=1) - p(D=1|E=0)
\]
and
\[
\widetilde{BF}_e^{\dagger} = p(E=1-e) p(D=1|E=e) + p(E=e)(1-p(D=1|E=1-e)).
\]

\section*{Appendix D: Manski's Sensitivity Analysis}

\citet{Manski1990} bounds $RD^{true}$ under the assumption that $D_0$ and $D_1$ take values in known intervals. The bounds apply to non-binary outcomes. So, we momentarily drop the assumption that $D$ is binary. \citeauthor{Manski1990}'s bounds are derived as follows. Suppose it is known that $D_1$ takes value in the interval $[K_{10}, K_{11}]$. Then,
\[
K_{10} \leq \mathbb{E}[D_1|E=0] \leq K_{11}.
\]
Consequently,
\[
\mathbb{E}[D_1] = \mathbb{E}[D_1|E=0] p(E=0) + \mathbb{E}[D_1|E=1] p(E=1)
\]
can be bounded as
\begin{align*}
K_{10} p(E=0) + \mathbb{E}[D|E=1] p(E=1) &\leq \mathbb{E}[D_1]\\
&\leq K_{11} p(E=0) + \mathbb{E}[D|E=1] p(E=1)
\end{align*}
by counterfactual consistency. Analogous bounds can be derived for $\mathbb{E}[D_0]$ under the assumption that $D_0$ takes values within the interval $[K_{00}, K_{01}]$. Consequently,
\[
RD^{true} = \mathbb{E}[D_1] - \mathbb{E}[D_0]
\]
can be bounded as
\begin{align}\label{eq:Manski}
K_{10} p(E=0) + \mathbb{E}[D|E=1] p(E=1)\\\nonumber
- \mathbb{E}[D|E=0] p(E=0) - K_{01} p(E=1) &\leq RD^{true}\\\nonumber
&\leq K_{11} p(E=0) + \mathbb{E}[D|E=1] p(E=1)\\\nonumber
&- \mathbb{E}[D|E=0] p(E=0) - K_{00} p(E=1).
\end{align}

When the outcome is binary, as in this work, $D_0$ and $D_1$ are definitionally bounded with $K_{00}=K_{10}=0$ and $K_{01}=K_{11}=1$, and the bounds take a simpler form:
\begin{align}\label{eq:Manski2}
p(D=1|E=1) p(E=1)\\\nonumber
- p(D=1|E=0) p(E=0) - p(E=1) &\leq RD^{true}\\\nonumber
& \leq p(E=0) + p(D=1|E=1) p(E=1)\\\nonumber
&- p(D=1|E=0) p(E=0).
\end{align}
Note that these bounds coincide with AS' bounds (Equation \ref{eq:bounds2AS}) and, thus, with our bounds when $M=1$ and $m=0$ (Equation \ref{eq:bounds2}).

Note also that, when $D$ is binary, Equations \ref{eq:bounds2} and \ref{eq:Manski} coincide if we let $M=K_{01}=K_{11}$ and $m=K_{00}=K_{10}$. Therefore, one may say that our bounds are an adaptation of Manski's bounds in Equation \ref{eq:Manski} to binary outcomes. An adaptation that retains the sensitivity parameters (unlike the direct application of Manski's bounds to binary outcomes in Equation \ref{eq:Manski2}) albeit with a different meaning (they now bound $p(D|E,U)$ rather than the support of $D_0$ and $D_1$). Retaining the sensitivity parameters is important because, recall, it is thanks to these parameters that our bounds can be made tighter than AS' and, thus, than those in Equation \ref{eq:Manski2}.

\end{document}